\documentclass[12pt,english]{article}
\usepackage[T1]{fontenc}
\usepackage[latin9]{inputenc}
\usepackage{color}
\usepackage{babel}
\usepackage{units}
\usepackage{amsmath}
\usepackage{amsthm}
\usepackage{amssymb}
\usepackage{graphicx}
\usepackage{graphics}
\usepackage[unicode=true,pdfusetitle,
 bookmarks=true,bookmarksnumbered=false,bookmarksopen=false,
 breaklinks=false,pdfborder={0 0 1},backref=page,colorlinks=true]
 {hyperref}
\hypersetup{
 linkcolor=blue, citecolor=blue, urlcolor=blue, filecolor=blue}

\makeatletter

\providecommand{\tabularnewline}{\\}

\theoremstyle{plain}
\newtheorem{thm}{\protect\theoremname}[section]
\theoremstyle{plain}
\newtheorem{prop}[thm]{\protect\propositionname}
\theoremstyle{definition}
\newtheorem{defn}[thm]{\protect\definitionname}
\theoremstyle{remark}
\newtheorem{rem}[thm]{\protect\remarkname}
\theoremstyle{plain}
\newtheorem{lem}[thm]{\protect\lemmaname}
\ifx\proof\undefined
\newenvironment{proof}[1][\protect\proofname]{\par
	\normalfont\topsep6\p@\@plus6\p@\relax
	\trivlist
	\itemindent\parindent
	\item[\hskip\labelsep\scshape #1]\ignorespaces
}{%
	\endtrivlist\@endpefalse
}
\providecommand{\proofname}{Proof}
\fi
\theoremstyle{definition}
\newtheorem{example}[thm]{\protect\examplename}

\usepackage{datetime,color,currfile,pdfsync,graphics,lineno,array}

\makeatother

\providecommand{\definitionname}{Definition}
\providecommand{\examplename}{Example}
\providecommand{\lemmaname}{Lemma}
\providecommand{\propositionname}{Proposition}
\providecommand{\remarkname}{Remark}
\providecommand{\theoremname}{Theorem}

\begin{document}

\title{Energy Functions in Polymer Dynamics in terms of Generalized Grey
Brownian Motion}

\author{\textbf{Wolfgang Bock}\\
Technomathematics Group\\
University of Kaiserslautern\\
P.\ O.\ Box 3049, 67653 Kaiserslautern, Germany\\
Email: bock@mathematik.uni-kl.de\and\textbf{ José Luís da Silva}\\
 CIMA, University of Madeira, Campus da Penteada,\\
 9020-105 Funchal, Portugal.\\
 Email: luis@uma.pt\and \textbf{Ludwig Streit}\\
BiBoS, Universität Bielefeld, Germany,\\
CIMA, Unversidade da Madeira, Funchal, Portugal\\
 Email: streit@uma.pt}
\maketitle
\begin{abstract}
In this paper we investigate the energy functions for a class of non
Gaussian processes. These processes are characterized in terms of
the Mittag-Leffler function. We obtain closed analytic form for the
energy function, in particular we recover the Brownian and fractional
Brownian motion cases.
\end{abstract}
\tableofcontents{}

\section{Introduction}

Polymers are consisting of small chemical units which act on each
other via different forces. A very simple and well studied model of
a homo-polymer, i.e.\ a polymer consisting of the same microunits
is a classical random walk. In that case it is known, that the nearest
neighbours are linked via springs, i.e. the chain can be considered
as a chain harmonic of oscillators. 

Of course to obtain a more realistic polymer model the suppression
of such walk had to be introduced (``excluded volume''), see \cite{Edwards-65},
\cite{Westwater1980} for a continuum model and for random walks,
see \cite{DOMB1972} and references therein. Individual chain polymer
models are hence well studied and widely understood. A continuum limit
of those models, i.e.\ where the polymers are modeled by Brownian
motion (Bm) paths, led to a deeper understanding in the asymptotic
scaling behavior of the chains. The draw-back of Brownian or random
walk models is that they can not reflect long-range forces along the
chain without introducing a further potential.

Fractional Brownian motion (fBm) paths have been suggested as a heuristic
model \cite{Biswas}, without yet including the ``excluded volume
effect'' although a more proper model would be based on self-avoiding
fractional random walks. 

The aim of this paper is first to understand the long-range correlations
of fBm as a generalized bead-spring model, hence a chain model. For
this we fist consider a continuous model, which we then discretize.
The aim of this paper is to go beyond the fBm based models. Hence
we discuss energy functions that arise from non-Gaussian chain models,
where the interaction potentials are not only long-range along the
chain but can also give rise to multi-particle non-linear forces between
the constituents. 

The class of underlying random processes is that of generalized grey
Brownian motion. 

In Section \ref{sec:ggBm} we shall introduce the required concepts
and properties of ggBm, so as to then present our results and in Section\ \ref{sec:energy-function}
the explicit energy functions are derived.

\section{Generalized Grey Brownian Motion in Arbitrary Dimensions}

\label{sec:ggBm}

\subsection{Construction of the Mittag-Leffler Measure }

\label{subsec:Prerequisites}Let $d\in\mathbb{N}$ and $L_{d}^{2}$
be the Hilbert space of vector-valued square integrable measurable
functions
\[
L_{d}^{2}:=L^{2}(\mathbb{R})\otimes\mathbb{R}^{d}.
\]
The space $L_{d}^{2}$ is unitary isomorphic to a direct sum of $d$
identical copies of $L^{2}:=L^{2}(\mathbb{R})$\textendash the space
of real-valued square integrable measurable functions wrt Lebesgue
measure. Any element $f\in L_{d}^{2}$ may be written in the form
\begin{equation}
f=(f_{1}\otimes e_{1},\ldots,f_{d}\otimes e_{d}),\label{eq:L2d_element}
\end{equation}
where $f_{i}\in L^{2}(\mathbb{R})$, $i=1,\ldots,d$ and $\{e_{1},\ldots,e_{d}\}$
denotes the canonical basis of $\mathbb{R}^{d}$. The norm of $f$
is given by
\[
|f|_{0}^{2}:=\sum_{k=1}^{d}|f_{k}|_{L^{2}}^{2}=\sum_{k=1}^{d}\int_{\mathbb{R}}f_{k}^{2}(x)\,dx.
\]

As a densely imbedded nuclear space in $L_{d}^{2}$ we choose $S_{d}:=S(\mathbb{R})\otimes\mathbb{R}^{d}$,
where $S(\mathbb{R})$ is the Schwartz test function space. An element
$\varphi\in S_{d}$ has the form
\begin{equation}
\varphi=(\varphi_{1}\otimes e_{1},\ldots,\varphi_{d}\otimes e_{d}),\label{eq:test_function.}
\end{equation}
where $\varphi_{i}\in S(\mathbb{R})$, $i=1,\ldots,d$. Together with
the dual space $S'_{d}:=S'(\mathbb{R})\otimes\mathbb{R}^{d}$ we obtain
the basic nuclear triple
\[
S_{d}\subset L_{d}^{2}\subset S'_{d}.
\]
The dual pairing between $S'_{d}$ and $S_{d}$ is given as an extension
of the scalar product in $L_{d}^{2}$ by 
\[
\langle w,\varphi\rangle_{0}=\sum_{k=1}^{d}\langle w_{k},\varphi_{k}\rangle,
\]
for any $w=(w_{1}\otimes e_{1},\ldots,w_{d}\otimes e_{d})\in S'_{d}$
with $w_{i}\in S'(\mathbb{R})$, $i=1,\ldots,d$ and $\varphi$ as
in (\ref{eq:test_function.}).

Let $0<\alpha<2$ be given and define the operator $M_{-}^{\nicefrac{\alpha}{2}}$
on $S(\mathbb{R})$ by
\[
M_{-}^{\nicefrac{\alpha}{2}}\varphi:=\begin{cases}
K_{\nicefrac{\alpha}{2}}D_{-}^{-\nicefrac{(\alpha-1)}{2}}\varphi, & \alpha\in(0,1),\\
\varphi, & \alpha=1,\\
K_{\nicefrac{\alpha}{2}}I_{-}^{\nicefrac{(\alpha-1)}{2}}\varphi, & \alpha\in(1,2),
\end{cases}
\]
where the normalization constant $K_{\nicefrac{\alpha}{2}}:=\sqrt{\alpha\sin(\nicefrac{\alpha\pi}{2})\Gamma(\alpha)}$
and $D_{-}^{r}$, $I_{-}^{r}$ denote the left-side fractional derivative
and fractional integral of order $r$ in the sense of Riemann-Liouville,
respectively. We refer to \cite{SKM1993} or \cite{KST2006} for the
details and proofs of fractional calculus. It is possible to obtain
a larger domain of the operator $M_{-}^{\nicefrac{\alpha}{2}}$ in
order to include the indicator function $1\!\!1_{[0,t)}$ such that
$M_{-}^{\nicefrac{\alpha}{2}}1\!\!1_{[0,t)}\in L^{2}$, for the details
we refer to Appendix A in \cite{GJ15}. We have the following 
\begin{prop}[Corollary 3.5 in \cite{GJ15}]
For all $t,s\ge0$ and all $0<\alpha<2$ it holds that
\begin{equation}
\big(M_{-}^{\nicefrac{\alpha}{2}}1\!\!1_{[0,t)},M_{-}^{\nicefrac{\alpha}{2}}1\!\!1_{[0,s)}\big)_{L^{2}}=\frac{1}{2}\big(t^{\alpha}+s^{\alpha}-|t-s|^{\alpha}\big).\label{eq:alpha-inner-prod}
\end{equation}
\end{prop}
We now introduce two special functions needed later on, namely Mittag-Leffler
and $M$-Wright functions as well as their relations.

The Mittag-Leffler function was introduced by G.\ Mittag-Leffler
in a series of papers \cite{Mittag-Leffler1903,Mittag-Leffler1904,Mittag-Leffler1905}.
\begin{defn}[Mittag-Leffler function]
\begin{enumerate}
\item For $\beta>0$ the Mittag-Leffler function $E_{\beta}$ (MLf for short)
is defined as an entire function by the following series representation
\begin{equation}
E_{\beta}(z):=\sum_{n=0}^{\infty}\frac{z^{n}}{\Gamma(\beta n+1)},\quad z\in\mathbb{C},\label{eq:MLf}
\end{equation}
where $\Gamma$ denotes the Gamma function. Note the relation $E_{1}(z)=e^{z}$
for any $z\in\mathbb{C}$. 
\end{enumerate}
\end{defn}
The Wright function is defined by the following series representation
which converges in the whole complex $z$-plane
\[
W_{\lambda,\mu}(z):=\sum_{n=0}^{\infty}\frac{z^{n}}{n!\Gamma(\lambda n+\mu)},\quad\lambda>-1,\;\mu\in\mathbb{C}.
\]
An important particular case of the Wright function is the so called
$M$-Wright function $M_{\beta}$, $0<\beta\le1$ (in one variable)
defined by 
\[
M_{\beta}(z):=W_{-\beta,1-\beta}(-z)=\sum_{n=0}^{\infty}\frac{(-z)^{n}}{n!\Gamma(-\beta n+1-\beta)}.
\]
For the choice $\beta=\nicefrac{1}{2}$ the corresponding $M$-Wright
function reduces to the Gaussian density
\begin{equation}
M_{\nicefrac{1}{2}}(z)=\frac{1}{\sqrt{\pi}}\exp\left(-\frac{z^{2}}{4}\right).\label{eq:MWright_Gaussian}
\end{equation}

The MLf $E_{\beta}$ and the $M$-Wright are related through the Laplace
transform
\begin{equation}
\int_{0}^{\infty}e^{-s\tau}M_{\beta}(\tau)\,d\tau=E_{\beta}(-s).\label{eq:LaplaceT_MWf}
\end{equation}

The $\mathbb{M}$-Wright function with two variables $\mathbb{M}_{\beta}^{1}$
of order $\beta$ (1-dimension in space) is defined by
\begin{equation}
\mathbb{M}_{\beta}^{1}(x,t):=\mathbb{M}_{\beta}(x,t):=\frac{1}{2}t^{-\beta}M_{\beta}(|x|t^{-\beta}),\quad0<\beta<1,\;x\in\mathbb{R},\;t\in\mathbb{R}^{+}\label{eq:MWf-2variables}
\end{equation}
which is a probability density in $x$ evolving in time $t$ with
self-similarity exponent $\beta$. The following integral representation
for the $\mathbb{M}$-Wright is valid, see \cite{Mainardi:2003vn}.
\begin{equation}
\mathbb{M}_{\nicefrac{\beta}{2}}(x,t)=2\int_{0}^{\infty}\frac{e^{-\nicefrac{x^{2}}{4\tau}}}{\sqrt{4\pi\tau}}t^{-\beta}M_{\beta}(\tau t^{-\beta})\,d\tau,\quad0<\beta\leq1,\;x\in\mathbb{R}.\label{eq:subordination_MWf1}
\end{equation}
This representation is valid a more general form, see \cite[eq.~(6.3)]{Mainardi:2003vn},
but for our purpose it is sufficient in view of its generalization
for $x\in\mathbb{R}^{d}$. In fact, eq.\ (\ref{eq:subordination_MWf1})
may be extended to a general spatial dimension $d$ by the extension
of the Gaussian function, namely
\begin{equation}
\mathbb{M}_{\nicefrac{\beta}{2}}^{d}(x,t):=2\int_{0}^{\infty}\frac{e^{-\frac{1}{4\tau}|x|^{2}}}{(4\pi\tau)^{\nicefrac{d}{2}}}t^{-\beta}M_{\beta}(\tau t^{-\beta})\,d\tau,\quad x\in\mathbb{R}^{d},\;t\ge0,\;0<\beta\le1.\label{eq:MWf_d_dimension}
\end{equation}
The function $\mathbb{M}_{\nicefrac{\beta}{2}}^{d}$ is nothing but
the density of the fundamental solution of a time-fractional diffusion
equation, see \cite{Mentrelli2015}.

The Mittag-Leffler measures $\mu_{\beta}$, $0<\beta\leq1$ are a
family of probability measures on $S'_{d}$ whose characteristic functions
are given by the Mittag-Leffler functions. On $S'_{d}$ we choose
the Borel $\sigma$-algebra $\mathcal{B}$ generated by the cylinder
sets, that is 
\[
\mathcal{F}C_{b}^{\infty}(S'_{d}):=\left\{ f(l_{1},\dots,l_{n})\,|\,n\in\mathbb{N},\,f\in C_{b}^{\infty}(\mathbb{R}^{n}),\,l_{1},\dots,l_{n}\in S'_{d}\right\} ,
\]
where $C_{b}^{\infty}(\mathbb{R}^{n})$ is the space of bounded infinitely
often differentiable functions on $\mathbb{R}^{n}$, where all partial
derivatives are also bounded. Using the Bochner-Minlos theorem, see
\cite{BK88} or \cite{HKPS93}, the following definition makes sense. 
\begin{defn}[{cf.\ \cite[Def.~2.5]{GJ15}}]
For any $\beta\in(0,1]$ the Mittag-Leffler measure is defined as
the unique probability measure $\mu_{\beta}$ on $S'_{d}$ by fixing
its characteristic functional
\begin{equation}
\int_{S'_{d}}e^{i\langle w,\varphi\rangle_{0}}\,d\mu_{\beta}(w)=E_{\beta}\left(-\frac{1}{2}|\varphi|_{0}^{2}\right),\quad\varphi\in S_{d}.\label{eq:ch-fc-gnm}
\end{equation}
\end{defn}
\begin{rem}
\label{rem:grey-noise-measure}
\begin{enumerate}
\item The measure $\mu_{\beta}$ is also called grey noise (reference) measure,
cf.\ \cite{GJRS14} and \cite{GJ15}.
\item The range $0<\beta\leq1$ ensures the complete monotonicity of $E_{\beta}(-x)$,
see Pollard \cite{Pollard48}, i.e., $(-1)^{n}E_{\beta}^{(n)}(-x)\ge0$
for all $x\ge0$ and $n\in\mathbb{N}_{0}:=\{0,1,2,\ldots\}.$ In other
words, this is sufficient to show that 
\[
S_{d}\ni\varphi\mapsto E_{\beta}\left(-\frac{1}{2}|\varphi|_{0}^{2}\right)\in\mathbb{R}
\]
is the characteristic function of a measure in $S'_{d}$. 
\end{enumerate}
\end{rem}
We consider the Hilbert space of complex square integrable measurable
functions defined on $S'_{d}$, 
\[
L^{2}(\mu_{\beta}):=L^{2}(S'_{d},\mathcal{B},\mu_{\beta}),
\]
with scalar product defined by
\[
(\!(F,G)\!)_{L^{2}(\mu_{\beta})}:=\int_{S'_{d}}F(w)\bar{G}(w)\,d\mu_{\beta}(w),\quad F,G\in L^{2}(\mu_{\beta}).
\]
The corresponding norm is denoted by $\|\cdot\|_{L^{2}(\mu_{\beta})}$. 

It follows from (\ref{eq:ch-fc-gnm}) that all moments of $\mu_{\beta}$
exists and we have
\begin{lem}
\label{lem:gnm}For any $\varphi\in S_{d}$ and $n\in\mathbb{N}$
we have
\begin{align*}
\int_{S'_{d}}\langle w,\varphi\rangle_{0}^{2n+1}\,d\mu_{\beta}(w) & =0,\\
\int_{S'_{d}}\langle w,\varphi\rangle_{0}^{2n}\,d\mu_{\beta}(w) & =\frac{(2n)!}{2^{n}\Gamma(\beta n+1)}|\varphi|_{0}^{2n}.
\end{align*}
In particular, $\|\langle\cdot,\varphi\rangle\|_{L^{2}}^{2}=\frac{1}{\Gamma(\beta+1)}|\varphi|_{0}^{2}$
and by polarization for any $\varphi,\psi\in S_{d}$ we obtain
\[
\int_{S'_{d}}\langle w,\varphi\rangle_{0}\langle w,\psi\rangle_{0}\,d\mu_{\beta}(w)=\frac{1}{\Gamma(\beta+1)}\langle\varphi,\psi\rangle_{0}.
\]
\end{lem}

\subsection{Generalized Grey Brownian Motion}

\label{subsec:ggBm}For any test function $\varphi\in S_{d}$ we define
the random variable
\[
X^{\beta}(\varphi):S'_{d}\longrightarrow\mathbb{R}^{d},\;w\mapsto X^{\beta}(\varphi,w):=\big(\langle w_{1},\varphi_{1}\rangle,\ldots,\langle w_{d},\varphi_{d}\rangle\big).
\]
The random variable $X^{\beta}(\varphi)$ has the following properties
which are a consequence of Lemma\ \ref{lem:gnm} and the characteristic
function of $\mu_{\beta}$ given in (\ref{eq:ch-fc-gnm}).
\begin{prop}
Let $\varphi,\psi\in S_{d}$, $k\in\mathbb{R}^{d}$ be given. Then
\begin{enumerate}
\item The characteristic function of $X^{\beta}(\varphi)$ is given by 
\begin{equation}
\mathbb{E}\big(e^{i(k,X^{\beta}(\varphi))}\big)=E_{\beta}\left(-\frac{1}{2}\sum_{j=1}^{d}k_{j}^{2}|\varphi_{j}|_{L^{2}}^{2}\right).\label{eq:characteristic-coord-proc}
\end{equation}
\item The characteristic function of the random variable $X^{\beta}(\varphi)-X^{\beta}(\psi)$
is
\begin{equation}
\mathbb{E}\big(e^{i(k,X^{\beta}(\varphi)-X^{\beta}(\psi))}\big)=E_{\beta}\left(-\frac{1}{2}\sum_{i=1}^{d}k_{j}^{2}|\varphi_{j}-\psi_{j}|_{L^{2}}^{2}\right).\label{eq:CF_increment}
\end{equation}
\item The moments of $X^{\beta}(\varphi)$ are given by
\begin{align}
\int_{S'_{d}}\big|X^{\beta}(\varphi,w)\big|^{2n+1}\,d\mu_{\beta}(w) & =0,\nonumber \\
\int_{S'_{d}}\big|X^{\beta}(\varphi,w)\big|^{2n}\,d\mu_{\beta}(w) & =\frac{(2n)!}{2^{n}\Gamma(\beta n+1)}|\varphi|_{0}^{2n}.\label{eq:variance.cood-proc}
\end{align}
\end{enumerate}
\end{prop}
The property (\ref{eq:variance.cood-proc}) of $X^{\beta}(\varphi)$
gives the possibility to extend the definition of $X^{\beta}$ to
any element in $L_{d}^{2}$, in fact, if $f\in L_{d}^{2}$, then there
exists a sequence $(\varphi_{k})_{k=1}^{\infty}\subset S_{d}$ such
that $\varphi_{k}\longrightarrow f$, $k\rightarrow\infty$ in the
norm of $L_{d}^{2}$. Hence, the sequence $\big(X^{\beta}(\varphi_{k})\big)_{k=1}^{\infty}\subset L^{2}(\mu_{\beta})$
forms a Cauchy sequence which converges to an element denoted by $X^{\beta}(f)\in L^{2}(\mu_{\beta})$. 

So, defining $1\!\!1_{[0,t)}\in L_{d}^{2}$, $t\ge0$, by
\[
1\!\!1_{[0,t)}:=(1\!\!1_{[0,t)}\otimes e_{1},\ldots,1\!\!1_{[0,t)}\otimes e_{d})
\]
we may consider the process $X^{\beta}(1\!\!1_{[0,t)})\in L^{2}(\mu_{\beta})$
such that the following definition makes sense.
\begin{defn}
For any $0<\alpha<2$ we define the process 
\begin{align}
S'_{d}\ni w\mapsto B^{\beta,\alpha}(t,w) & :=\big(\langle w,(M_{-}^{\nicefrac{\alpha}{2}}1\!\!1_{[0,t)})\otimes e_{1}\rangle,\ldots,\langle w,(M_{-}^{\nicefrac{\alpha}{2}}1\!\!1_{[0,t)})\otimes e_{d}\rangle\big)\nonumber \\
 & =\big(\langle w_{1},M_{-}^{\nicefrac{\alpha}{2}}1\!\!1_{[0,t)}\rangle,\ldots,\langle w_{d},M_{-}^{\nicefrac{\alpha}{2}}1\!\!1_{[0,t)}\rangle\big),\;t>0\label{eq:ggBm}
\end{align}
as an element in $L^{2}(\mu_{\beta})$ and call this process $d$-dimensional
generalized grey Brownian motion (ggBm for short). Its characteristic
function has the form
\begin{equation}
\mathbb{E}\big(e^{i(k,B^{\beta,\alpha}(t))}\big)=E_{\beta}\left(-\frac{|k|^{2}}{2}t^{\alpha}\right),\;k\in\mathbb{R}^{d}.\label{eq:ch-fc-ggBm}
\end{equation}
\end{defn}
\begin{rem}
\label{rem:existence_ggBm}
\begin{enumerate}
\item The $d$-dimensional ggBm $B^{\beta,\alpha}$ exist as a $L^{2}(\mu_{\beta})$-limit
and hence the map $S'_{d}\ni\omega\mapsto\langle\omega,M_{-}^{\alpha/2}1\!\!1_{[0,t)}\rangle$
yields a version of ggBm, $\mu_{\beta}$-a.s., but not in the pathwise
sense. 
\item For any fixed $0<\alpha<2$ one can show by the Kolmogorov-Centsov
continuity theorem that the paths of the process are $\mu_{\beta}$-a.s.\ continuous,
cf.\ \cite[Prop.~3.8]{GJ15}.
\item Below we mainly deal with expectation of functions of ggBm, therefore
the version of ggBm defined above is sufficient.
\end{enumerate}
\end{rem}
\begin{prop}
\label{prop:characterization_ggBm}
\begin{enumerate}
\item For any $0<\alpha<2$, the process $B^{\beta,\alpha}:=\{B^{\beta,\alpha}(t),\;t\ge0\}$,
is $\nicefrac{\alpha}{2}$-self-similar with stationary increments. 
\item The finite dimensional probability density functions are given for
any $0\leq t_{1}<t_{2}<\ldots<t_{n}<\infty$ by
\[
\rho_{n}^{\beta,\alpha}(x,Q)=\frac{1}{(2\pi)^{\frac{dn}{2}}\big(\det(Q)\big)^{\frac{1}{2}}}\int_{0}^{\infty}\frac{1}{\tau^{\frac{dn}{2}}}e^{-\frac{\|x\|_{Q}^{2}}{4\tau}}M_{\beta}(\tau)\,d\tau,\quad x\in\mathbb{R}^{dn},
\]
where $Q=(a_{ij})$ is the covariance matrix given by 
\[
a_{ij}=\mathbb{E}\big((B^{\beta,\alpha}(t_{i}),B^{\beta,\alpha}(t_{j}))\big)=\frac{d}{2\Gamma(\beta+1)}\left(t_{i}^{\alpha}+t_{j}^{\alpha}-|t_{i}-t_{j}|^{\alpha}\right)
\]
and $\|x\|_{Q}^{2}:=(x,Q^{-1}x)_{\mathbb{R}^{dn}}$.
\end{enumerate}
\end{prop}
\begin{proof}
The proof can be found in \cite{Mura_Pagnini_08}.
\end{proof}
\begin{rem}
\label{rem:self-similar}The family $\{B^{\beta,\alpha}(t),\;t\ge0,\,\beta\in(0,1],\,\alpha\in(0,2)\}$
forms a class of $\nicefrac{\alpha}{2}$-self-similar processes with
stationary increments ($\nicefrac{\alpha}{2}$-sssi) which includes:
\begin{enumerate}
\item For $\beta=\alpha=1$, the process $\{B^{1,1}(t),\;t\ge0\}$ is a
standard $d$-dimensional Bm.
\item For $\beta=1$ and $0<\alpha<2$, $\{B^{1,\alpha}(t),\;t\ge0\}$ is
a $d$-dimensional fBm with Hurst parameter $\nicefrac{\alpha}{2}$.
\item For $\alpha=1$, $\{B^{\beta,1}(t),\;t\ge0\}$ is a $\nicefrac{1}{2}$-self-similar
non Gaussian process with
\begin{equation}
\mathbb{E}\left(e^{i\left(k,B^{\beta,1}(t)\right)}\right)=E_{\beta}\left(-\frac{|k|^{2}}{2}t\right),\quad k\in\mathbb{R}^{d}.\label{eq:ch-fc-1/2sssi}
\end{equation}
\item For $0<\alpha=\beta<1$, the process $\{B^{\beta}(t):=B^{\beta,\beta}(t),\;t\ge0\}$
is $\nicefrac{\beta}{2}$-self-similar and is called $d$-dimensional
grey Brownian motion (gBm for short). Its characteristic function
is given by 
\begin{equation}
\mathbb{E}\left(e^{i\left(k,B^{\beta}(t)\right)}\right)=E_{\beta}\left(-\frac{|k|^{2}}{2}t^{\beta}\right),\quad k\in\mathbb{R}^{d}.\label{eq:ch-fc-gBm}
\end{equation}
For $d=1$, this process was introduced by W.\ Schneider in \cite{Schneider90,MR1190506}.
\item For other choices of the parameters $\beta$ and $\alpha$ we obtain,
in general, non Gaussian processes.
\end{enumerate}
\end{rem}

\section{Energy Functions}

\label{sec:energy-function}In this section we compute the energy
functions (also called Hamiltonian) associated to the system driven
by a ggBm. At first, we point out the classical case driven by a Bm
which is the sum of harmonic oscillators potentials corresponding
to nearest neighbors attraction. We then compute the energy function
for the general non-Gaussian process $B^{\beta,\alpha}$.

\subsection{Gaussian Case}

Let $X=\{X(t),\;t\ge0\}$ be a standard Gaussian process in $\mathbb{R}^{d}$
with covariance 
\[
\mathbb{E}\big(X_{i}(t)X_{j}(s)\big)=R^{X}(t,s)\delta_{ij}.
\]
Denote the discrete increments of $X$ by
\[
Y(k):=X(k)-X(k-1),\quad k=1,\ldots,N,\;N\in\mathbb{N}.
\]
The density of the $\mathbb{R}^{dN}$-valued random variable $Y=\big(Y(1),\ldots,Y(N)\big)$
may be computed from its characteristic function, namely for any $\lambda=(\lambda_{1},\ldots,\lambda_{N})\in\mathbb{R}^{dN}$
\begin{align*}
C(\lambda):=\mathbb{E}\big(e^{i(Y,\lambda)}\big) & =\mathbb{E}\left(\exp\left(i\sum_{k=1}^{N}\big(Y(k),\lambda_{k}\big)_{\mathbb{R}^{d}}\right)\right)\\
 & =\int_{\mathbb{R}^{dN}}\varrho_{N}^{X}(y)\exp\left(i\sum_{k=1}^{N}\big(y_{k},\lambda_{k}\big)_{\mathbb{R}^{d}}\right)dy.
\end{align*}
If we represent this characteristic function $C$ by
\[
C(\lambda)=\int_{\mathbb{R}^{dN}}e^{-H^{X}(y)}\exp\left(i\sum_{k=1}^{N}\big(y_{k},\lambda_{k}\big)_{\mathbb{R}^{d}}\right)dy,
\]
then by an inverse Fourier transform we obtain
\[
H^{X}(y)=-\ln\big(\varrho^{X}(y)\big),\quad y\in\mathbb{R}^{dN}.
\]
The function $H^{X}$ is called energy function (or Hamiltonian) of
the system. 
\begin{example}
\label{exa:hamiltonian}
\begin{enumerate}
\item If $X$ is the Brownian motion $B$ in $\mathbb{R}^{d}$, putting
$t_{k}=k$, $k=1,\ldots,N$, up to an irrelevant constant, the function
$H^{B}$ is given by
\[
H^{B}(y)=\frac{1}{2}\sum_{k=1}^{N}|y_{k}|^{2}=\frac{1}{2}\sum_{k=1}^{N}|x_{k}-x_{k-1}|^{2},
\]
where $x_{k}$, $k=1,\ldots,N$ denotes the integrated variables. 
\item For the fractional Brownian motion $X=B^{h}$ with Hurst parameter
$h\in(0,1)$ we find up to an irrelevant constant 
\[
H^{B^{h}}(y)=\frac{1}{2}(y,\Sigma^{-1}y)=\sum_{k,n=1}^{N}y_{k}\sigma_{kn}y_{n}=\sum_{k,n=1}^{N}g_{kn}(x_{k}-x_{n})^{2},
\]
where $\Sigma^{-1}=(\sigma_{kn})_{k,n=1}^{N}$ denotes the inverse
of the covariance matrix of the increments $Y^{h}(k)=B^{h}(k)-B^{h}(k-1)$,
$k=1,\ldots,N$.
\end{enumerate}
\end{example}
\begin{rem}
\label{rem:H_Bm-case}
\begin{enumerate}
\item Note that the terms
\[
V=\frac{1}{2}|x_{k}-x_{k-1}|^{2}
\]
are harmonic oscillator potentials, attracting nearest neighbors.
Thus, for Gaussian processes, $H^{X}$ may be calculated $H^{X}=-\ln\varrho^{X}$
through the characteristic function of the process by an inverse Fourier
transform. In addition, $H^{X}$ will be always a quadratic form,
basically the inverse of the covariance matrix $a_{kl}:=\mathbb{E}\big((Y(k),Y(l))_{\mathbb{R}^{d}}\big)$,
$k,l=1,\ldots,N$. 
\item For the fBm case, the difference is that the interaction is not anymore
confined to nearest neighbors. For small Hurst index, this inverted
matrix leads to a long-range attraction along the chain making it
curlier than (discretized) Brownian, while for Hurst indices larger
than $\nicefrac{1}{2}$ there appears a repulsion of next-to-nearest
neighbors, stretching the chain, see \cite{BBS18} and Figure\ \ref{fig:coeff_osc_fBm}
which displays the coupling $g_{kn}$ between between the central
particle and the others along the chain. For $H=0.3$ and $H=0.8$. 
\end{enumerate}
\end{rem}
\begin{figure}
	\begin{centering}
\includegraphics[width=.9\textwidth]{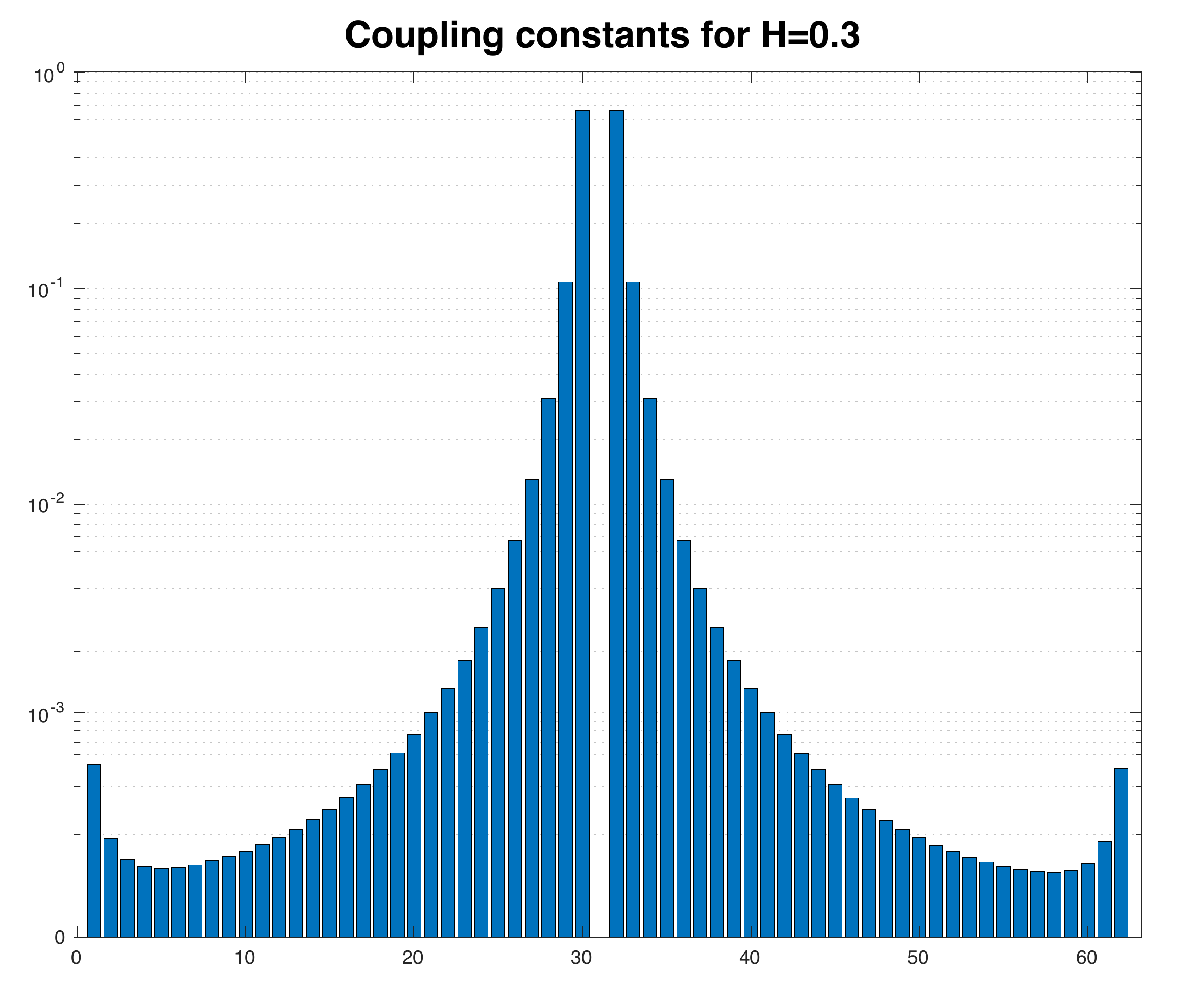}\hfill{}
\includegraphics[width=.9\textwidth]{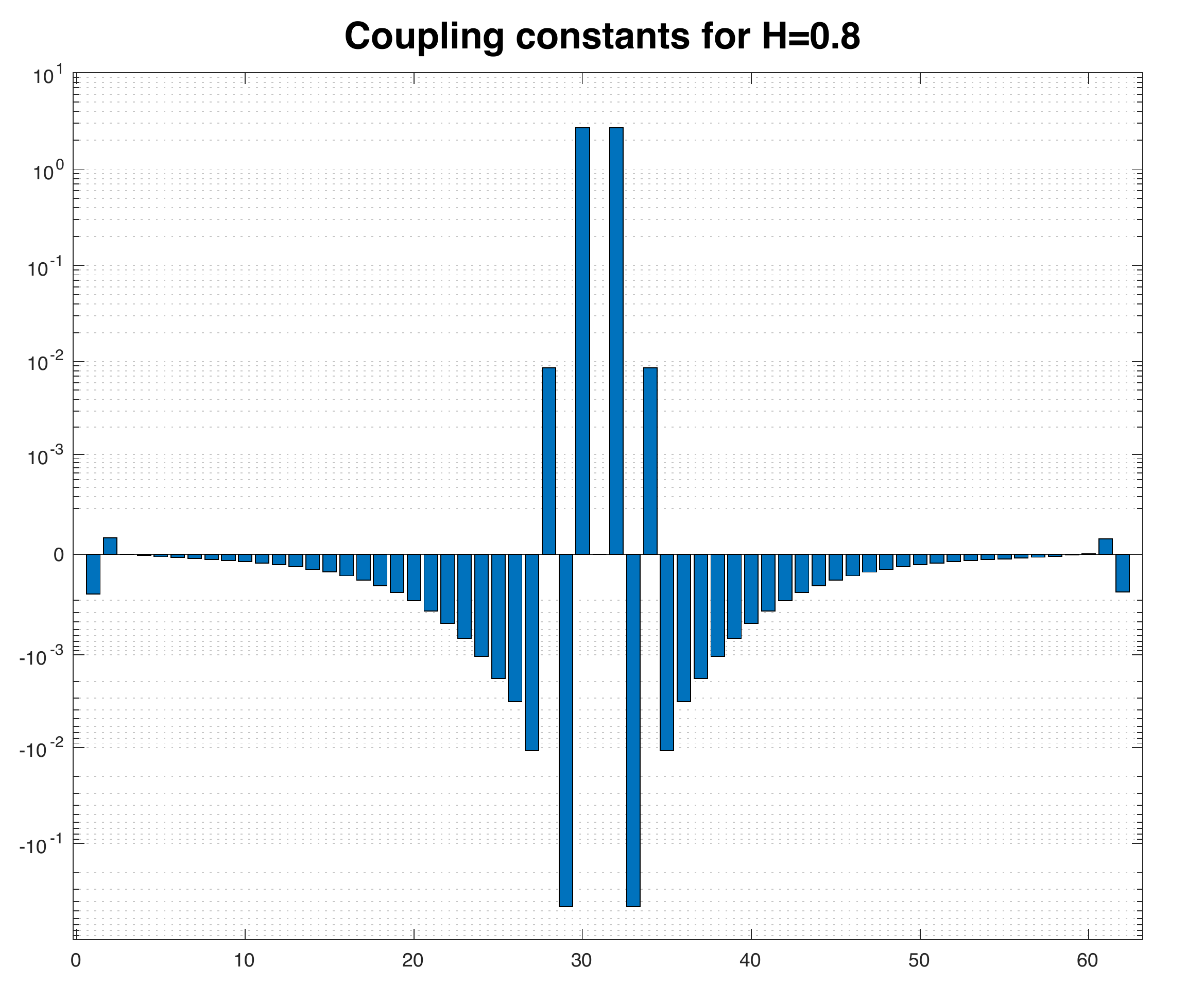} 
\caption{\label{fig:coeff_osc_fBm}Coupling constants for small Hurst index
 and large Hurst index \cite{BBS18}.}
\end{centering}
\end{figure}

\subsection{A non-Gaussian generalization}

In general, for the non Gaussian case, the energy function will no
more be quadratic. To keep things simple let us for the moment just
look at the case of $N=2$. 

\subsubsection{Energy function for $2$ particle interaction}

We look at the increment $Y_{kl}=B^{\beta,\alpha}(k)-B^{\beta,\alpha}(l)$
for $0<l<k<\infty$. The function $H^{\beta,\alpha}(y)=-\ln\big(\varrho^{\beta,\alpha}(y)\big)$
can be computed from the characteristic function of $Y_{kl}$, i.e.,
for any $\lambda\in\mathbb{R}^{d}$ we have 
\begin{align*}
\mathbb{E}\left(e^{i(B^{\beta,\alpha}(k)-B^{\beta,\alpha}(l),\lambda)}\right) & =E_{\beta}\left(-\frac{1}{2}|\lambda|^{2}|k-l|^{\alpha}\right)\\
 & =\int_{\mathbb{R}^{d}}\varrho_{1}^{\beta,\alpha}(y)\exp\big(i(y,\lambda)\big)dy
\end{align*}
through the inverse Fourier transform with $\zeta:=|k-l|^{\alpha}$
\begin{align*}
\varrho_{1}^{\beta,\alpha}(y) & =\frac{1}{(2\pi)^{d}}\int_{\mathbb{R}^{d}}\exp\big(-i(y,\lambda)_{\mathbb{R}^{d}}\big)E_{\beta}\left(-\frac{1}{2}|\lambda|^{2}\zeta\right)d\lambda\\
 & =\frac{1}{(2\pi)^{d}}\int_{0}^{\infty}M_{\beta}(\tau)\int_{\mathbb{R}^{d}}\exp\left(-i(y,\lambda)_{\mathbb{R}^{d}}-\frac{1}{2}|\lambda|^{2}\zeta\tau\right)d\lambda\,d\tau.
\end{align*}
Computing the Gaussian integral and using equality (\ref{eq:MWf_d_dimension})
yields
\begin{align*}
\varrho_{1}^{\beta,\alpha}(y) & =\frac{1}{(2\pi)^{\nicefrac{d}{2}}}\int_{0}^{\infty}\frac{1}{(\tau\zeta)^{\nicefrac{d}{2}}}\exp\left(-\frac{1}{2\tau\zeta}|y|^{2}\right)M_{\beta}(\tau)\,d\tau\\
 & =2^{\nicefrac{d}{2}}\int_{0}^{\infty}\frac{1}{(4\pi\tau)^{\nicefrac{d}{2}}}\exp\left(-\frac{1}{4\tau}|\sqrt{2}y|^{2}\right)\big(\zeta^{\frac{1}{\beta}}\big)^{-\beta}M_{\beta}\big(\tau\big(\zeta^{\frac{1}{\beta}}\big)^{-\beta}\big)\,d\tau\\
 & =2^{\nicefrac{d}{2}-1}\mathbb{M}_{\nicefrac{\beta}{2}}^{d}\left(\sqrt{2}y,\zeta^{\frac{1}{\beta}}\right).
\end{align*}
Therefore, up to a constant the function $H^{\beta,\alpha}$ is given
by
\[
H^{\beta,\alpha}(y)=-\ln\big(\mathbb{M}_{\nicefrac{\beta}{2}}^{d}\big(\sqrt{2}y,\zeta^{\frac{1}{\beta}}\big)\big),\quad y\in\mathbb{R}^{d}.
\]

In Figure \ref{fig:Energy-functions} we plot $H^{\beta,\alpha}$
for different values of $\beta$ and $\alpha$ assuming a time length
$|k-l|=1$. 
\begin{figure}
\begin{centering}
\includegraphics[scale=0.75]{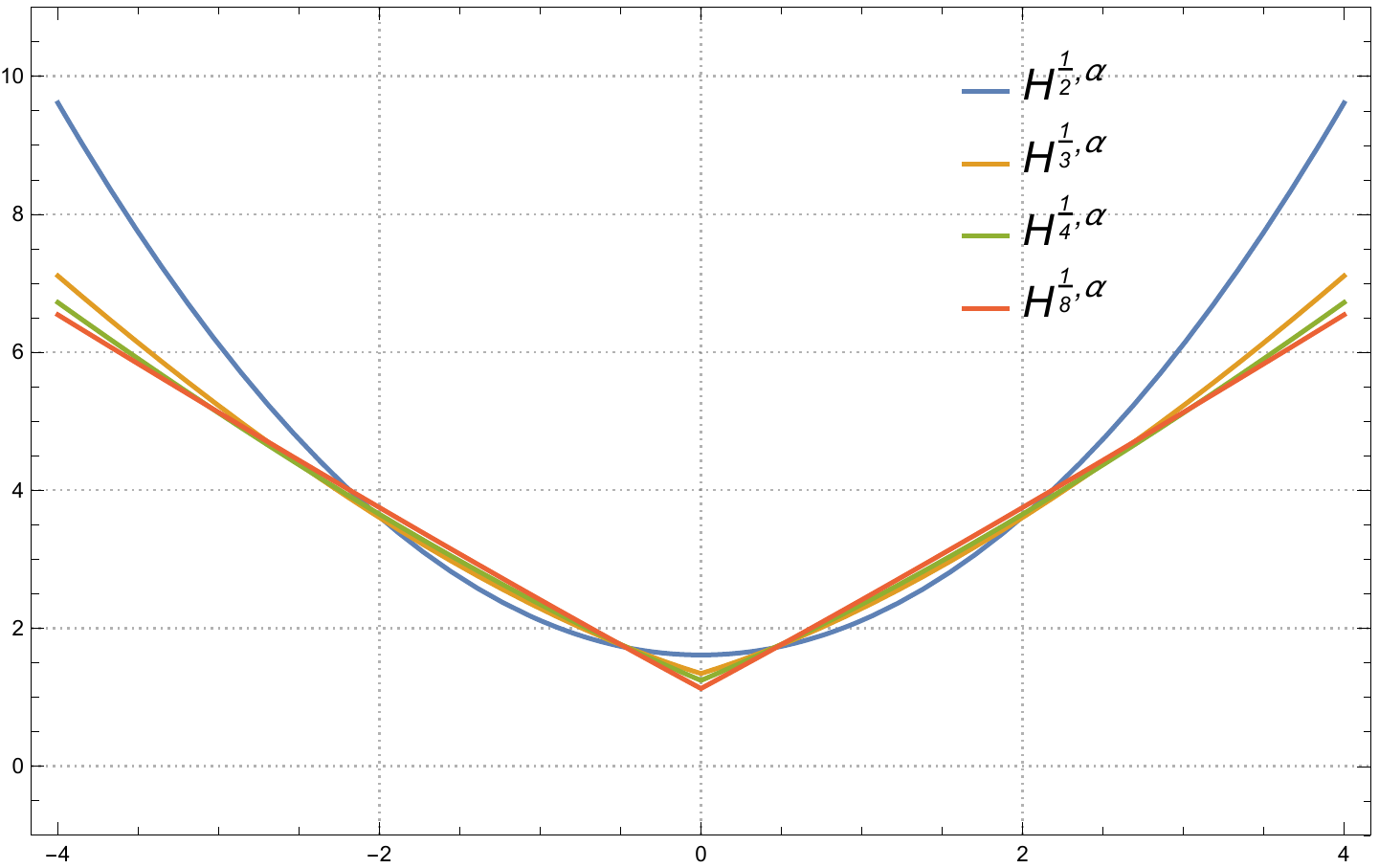} 
\par\end{centering}
\caption{\label{fig:Energy-functions}Energy functions $H^{\beta,\alpha}$
for $d=1$; $|k-l|=1$ and $\beta=1,\frac{2}{3},\frac{1}{2},\frac{1}{4}$.}
\end{figure}

\subsubsection{Energy function for $N+1$ particle integration}

In general, for an arbitrary $N\in\mathbb{N}$, $H^{\beta,\alpha}$
may be computed using the same technique, namely consider the vectors
\[
Y(k):=Y^{\beta,\alpha}(k):=B^{\beta,\alpha}(k)-B^{\beta,\alpha}(k-1),\quad k=1,\ldots,N.
\]
and compute the characteristic function of $Y=\big(Y^{\beta,\alpha}(1),\ldots,Y^{\beta,\alpha}(N)\big)$
for any $\lambda=(\lambda_{1},\ldots,\lambda_{N})\in(\mathbb{R}^{d})^{N}$
\[
\mathbb{E}\left(e^{i(Y,\lambda)_{\mathbb{R}^{dN}}}\right)=E_{\beta}\left(-\frac{1}{2}(\lambda,Q\lambda)_{\mathbb{R}^{dN}}\right),
\]
where $Q:=Q^{\beta,\alpha}=(a_{kn})_{k,n=1}^{N}$ is the covariance
matrix of $Y$ given by
\[
a_{kn}=\mathbb{E}\big((Y(k),Y(n))\big)=\frac{d}{2\Gamma(\beta+1)}\big[|k-1-n|^{\alpha}+|n-1-k|^{\alpha}-2|k-n|^{\alpha}\big].
\]
By Proposition\ (\ref{prop:characterization_ggBm})-2 the density
of $Y$ is given by
\[
\varrho_{N}^{\beta,\alpha}(y)=\frac{1}{(2\pi)^{\nicefrac{dN}{2}}\big(\det Q\big)^{\nicefrac{d}{2}}}\int_{0}^{\infty}\frac{1}{\tau^{\nicefrac{dN}{2}}}\exp\left(-\frac{1}{2\tau}\|y\|_{Q}^{2}\right)M_{\beta}(\tau)\,d\tau.
\]
Hence, the energy function is
\[
H^{\beta,\alpha}(y)=-\ln\left(\varrho_{N}^{\beta,\alpha}(y)\right).
\]

\begin{example}
Let us consider the special case of the 3-particle interaction in
dimension $d=1$. The previous result gives the energy function 
\[
\varrho_{3}^{\beta,\alpha}(y)=\frac{1}{(2\pi)^{d}\big(\det Q\big)^{\nicefrac{d}{2}}}\int_{0}^{\infty}\frac{1}{\tau^{d}}\exp\left(-\frac{1}{2\tau}\|y\|_{Q}^{2}\right)M_{\beta}(\tau)\,d\tau.
\]
For special values of the parameter $\beta$, we may compute in a
closed form the density $\varrho_{3}^{\beta,\alpha}$ as follows. 
\end{example}
\begin{center}
\begin{tabular}{|c|c|c|}
\hline 
 & $M_{\beta}(\tau)$ & $\varrho_{3}^{\beta,\alpha}(y)$\tabularnewline
\hline 
\hline 
$\beta=0$ & $e^{-\tau}$ & ${\displaystyle \frac{1}{\pi\det(Q)^{\nicefrac{1}{2}}}K_{0}\left(\sqrt{2}\|y\|_{Q}\right)}$\tabularnewline
\hline 
$\beta=\frac{1}{3}$ & ${\displaystyle 3^{\nicefrac{2}{3}}\mathrm{Ai}(3^{-\nicefrac{1}{3}}\tau)}$ & ${\displaystyle \frac{1}{8\pi^{3}\det(Q)^{\nicefrac{1}{2}}}G_{5,0}^{0,0}\left(\frac{\|y\|_{Q}^{6}}{5832}\bigg|\;\begin{smallmatrix}\hline \\ \\
0,0,\frac{1}{3},\frac{1}{3},\frac{2}{3}
\end{smallmatrix}\;\right)}$\tabularnewline
\hline 
$\beta=\frac{1}{2}$ & ${\displaystyle \frac{1}{\sqrt{\pi}}e^{-\frac{\tau^{2}}{4}}}$ & ${\displaystyle \frac{1}{(2\pi)^{2}\det(Q)^{\nicefrac{1}{2}}}G_{3,0}^{0,0}\left(\frac{\|y\|_{Q}^{4}}{64}\bigg|\;\begin{smallmatrix}\hline \\ \\
0,0,\frac{1}{2}
\end{smallmatrix}\;\right)}$\tabularnewline
\hline 
$\beta=1$ & $\delta_{1}(\tau)$ & ${\displaystyle \frac{1}{2\pi\det(Q)^{\nicefrac{1}{2}}}\exp\left(-\frac{1}{2}\|y\|_{Q}^{2}\right)}$\tabularnewline
\hline 
\end{tabular}
\par\end{center}

Here $K_{\nu}$, $G_{p,q}^{m,n}$ and $\mathrm{Ai}$ are the BesselK,
Meijer $G$ and Airy functions respectively, see \cite{Olver2010}.

The specific form of the higher order interactions which arise for
$\beta<1$, based on a Taylor expansion of $\ln(\varrho^{\beta,\alpha}(y)$
will be the subject of a separated investigation. 

\subsection*{Acknowledgement}

Financial support from FCT \textendash{} Funda{\c c\~a}o para a Ci{\^e}ncia
e a Tecnologia through the project UID/MAT/04674/2013 (CIMA Universidade
da Madeira) is gratefully acknowledged.

\bibliographystyle{alpha}

\end{document}